\documentclass[12pt]{article}

\usepackage[T1]{fontenc}
\usepackage[utf8]{inputenc}
\usepackage{amsthm,amsmath,amssymb}
\usepackage{graphicx}
\usepackage[colorlinks=true,citecolor=black,linkcolor=black,urlcolor=blue]{hyperref}

\usepackage{diagbox}
\usepackage{tikz}

\usepackage{color}
\usepackage{enumerate}
\usepackage{epsfig}
\usepackage{pgf}
\usepackage{pgfcore}
\usepackage{footnote}

\usepackage{diagbox} 

\theoremstyle{plain}
\newtheorem{theorem}{Theorem}
\newtheorem{lemma}[theorem]{Lemma}
\newtheorem{corollary}[theorem]{Corollary}
\newtheorem{proposition}[theorem]{Proposition}

\theoremstyle{definition}
\newtheorem{definition}[theorem]{Definition}

\definecolor{rougejoli}{RGB}{217.,1,21.}
\renewenvironment{proof}{
\medskip

\par \noindent \textbf{Proof.}
}{\hfill$\Box$\medskip}

\theoremstyle{remark}
\newtheorem{remark}[theorem]{Remark}

\tikzstyle{noeud}=[circle,inner sep=2, minimum size =3 pt, line width = 1pt, draw=black, fill=white]
\tikzstyle{petit_noeud}=[circle,inner sep=1.5, minimum size =2.5 pt, line width = 0.75pt, draw=black, fill=white]

\newcommand{\Dompos}{{\tt Dominator}}
\newcommand{\Stpos}{{\tt Staller}}
\newcommand{\Dom}{{Dominator}}
\newcommand{\St}{{Staller}}
\newcommand{\MB}{Maker-Breaker domination game}

\newcommand{\glue}[2]{
    \tikz[baseline = -5]{
    
        \clip (-0.3,-0.25) rectangle (0.3,0.28);
        \fill (0,0) circle (0.05) node[below left]{{\tiny $#1$}} node[below right]{{\tiny $#2$}};
        \draw (-0.2,-0.1) to[out =0, in = -90] (0,0) ;
        \draw (-0.2,0.1) to[out =0, in = 90] (0,0);
        \draw (0.2,-0.1) to[out =180, in = -90] (0,0) ;
        \draw (0.2,0.1) to[out =180, in = 90] (0,0);

    }
}

\title{Maker-Breaker domination game}
\author{Eric Duchêne $^{a}$ \and Valentin Gledel  $^{a}$ \and Aline Parreau  $^{a}$ \and Gabriel Renault $^b$}
\date{}

\begin{document}

\maketitle

\begin{center}
	$^a$ Univ Lyon, Universit\'e Lyon 1, LIRIS UMR CNRS 5205, F-69621, Lyon, France \\
	\medskip
	$^b$ UMONS, Mons, Belgium \\
		
\end{center}

\begin{abstract}
We introduce the Maker-Breaker domination game, a two player game on a graph. At his turn, the first player, Dominator, selects a vertex in order to dominate the graph while the other player, Staller, forbids a vertex to Dominator in order to prevent him to reach his goal. Both players play alternately without missing their turn. This game is a particular instance of the so-called Maker-Breaker games, that is studied here in a combinatorial context. In this paper, we first prove that deciding the winner of the Maker-Breaker domination game is {\sc pspace}-complete, even for bipartite graphs and split graphs. It is then showed that the problem is polynomial for cographs and trees. In particular, we define a strategy for Dominator that is derived from a variation of the dominating set problem, called the {\em pairing dominating set} problem. 
\end{abstract}


\medskip\noindent
{\bf AMS Subj.\ Class:} 05C57, 05C69, 91A43

\section{Introduction}

Since their introduction by Erd\H{o}s and Selfridge in \cite{erdos-1973}, positional games have been widely studied in the literature (see \cite{hefetz2014positional} for a recent survey book on the topic). These games are played on an hypergraph of vertex set $X$, with a finite set $\mathcal{F}\subseteq 2^X$ of hyperedges. The set $X$ is often called the {\em board} of the game, and an element of $\mathcal{F}$ a {\em winning set}. The game involves two players that alternately occupy a previously unoccupied vertex of $X$. The winner is determined by a convention: in the {\em Maker-Maker} convention, the first player to occupy all the vertices of a winning set is the winner. Such games may end in a draw, as it is the case in Tic-Tac-Toe. In the {\em Maker-Breaker} convention, the objectives are opposite: one player (the {\em Maker}) aims to occupy all the vertices of a winning set, whereas {\em Breaker} wins if she occupies a vertex in every winning set. In view of the complexity of  solving both kinds of games, Maker-Breaker instances are generally more considered in the literature as by definition, there is always a winner. In addition, rulesets of such games are often built from a graph. For example, one can mention the famous {\em Shannon switching game} (popularized as the game {\sc Bridg-it}) \cite{shannon}, where, given a graph $G=(V,E)$ and two particular vertices $u$ and $v$, the board $X$ corresponds to $E$, and the winning sets are all the subsets of $E$ that form a $u-v$ path in $G$. In the {\em Hamiltonicity} game \cite{connectivity}, the winning sets are all the sets of edges containing an Hamiltonian cycle.  \\

In view of such examples, converting a graph property into a $2$-player game is a natural operation. Hence it is not surprising that it has also been done for dominating sets. More precisely, several games having different rulesets and known as {\em domination game} have been defined in the literature. For example, in \cite{alon,favaron}, a move consists in orienting an edge of a given graph $G$ and the two players try to maximize (resp. minimize) the domination number of the resulting digraph. In \cite{Bicoldom}, the rules require two colors during the play. In \cite{bresar}, the domination game is defined in a sense where the players both select vertices and try to maximize (resp. minimize) the length of the play before building a dominating set. Since then, this version has become the standard one for the domination game, with regular progress on it \cite{dom1,dom2,dom3,dom4}. However, among the different variants of the domination game , the natural Maker-Breaker version (in the sense of Erd\H{o}s and Selfridge) has never been considered in the literature. In this paper, we consider the so-called {\em Maker Breaker Domination game}, where, given a graph $G=(V,E)$, the board $X$ is the set $V$, and $\mathcal{F}$ is the set of all the dominating sets of $G$. In other words, the two players alternately occupy a not yet occupied vertex of $G$. Maker wins if he manages to build a dominating set of $G$, whereas Breaker wins if she manages to occupy a vertex and all its neighbors. In what follows and in order to be consistent with the standard domination game, Maker will be called {\em Dominator}, and Breaker will be the {\em Staller}. \\



When dealing with Maker-Breaker games, there are two main questions that naturally arise:

\begin{itemize}
\item Given a graph $G$, which player has a winning strategy for the \MB\ on $G$ ?
\item If \Dom\ has a winning strategy on $G$, what is the minimum number of turns needed to win?
\end{itemize}

The current paper is about the first question. In the next section, we give definitions for the different cases about the winner, together with first general results. Section 3 deals with the  algorithmic complexity of the problem, where the {\sc pspace}-completeness is proved. In Section 4, a so-called {\em pairing strategy} is given, yielding a strategy for \Dom\ in graphs having certain properties. The last section is about graph operators that lead to polynomial strategies on trees and cographs.

\section{Preliminaries}


A \emph{position} of the \MB\ is denoted by a triplet $G=(V,E,c)$, where $V$ is a set of vertices, $E$ is a set of edges on $V$ and $c$ is a function $c:V \rightarrow \{\Dompos,\Stpos, {\tt Unplayed}\}$. In other words, the function $c$ allows to describe any game position encountered during the play. 
If, for all $u$ in $V$, $c(u)={\tt Unplayed}$, then $G$ is said to be a \emph{starting position}. In this case, we will identify $G$ with the graph $(V,E)$. At his turn, Dominator (respectively Staller) chooses one vertex $u$ with $c(u)={\tt Unplayed}$ and changes its value to $\Dompos$ (resp. $\Stpos$). When there is no more {\tt Unplayed} vertex, either the set of vertices $c^{-1}(\Dompos)$ forms a dominating set, and Dominator wins, or there is one vertex $u$ for which all its closed neighboorhood has value $\Stpos$, and Staller wins. In the latter case, we say that Staller  {\em isolates} $u$. Note that whenever $c^{-1}(\Dompos)$ is a dominating set or a vertex has been isolated by Staller, the winner is already determined and cannot change, since the two conditions are complementary. Thus we will often consider that the game stops when one of the two conditions holds.\\

The \MB\ is a finite game with perfect information and no draw. Thus, there is always a winning strategy for one of the player. There are four cases - also called {\em outcomes} - to characterize the winner of the game, according to who starts. We define $\mathcal D$, $\mathcal S$, $\mathcal N$ and $\mathcal P$ as the different possible outcomes for a position of the \MB . 

\begin{definition}
A position $G$ has four possible outcomes:
\begin{itemize}
\item $\mathcal D$ if Dominator has a winning strategy as first and second player,
\item $\mathcal S$ if Staller has a winning strategy as first and second player,
\item $\mathcal N$ if the next player (i.e., the one who starts) has a winning strategy,
\item $\mathcal P$ otherwise (i.e., the second player wins).
\end{itemize}
\end{definition}

Note that for proximity reasons, the notion of outcome and the last two notations are derived from combinatorial game theory \cite{Siegel}. In addition, the outcome of $G$ is denoted $o(G)$.\\

The following proposition is a direct application of a general result on Maker-Breaker games stated in \cite{beck2008combinatorial, hefetz2014positional}. It ensures that the outcome $\mathcal P$ never occurs. For the sake of completeness, we here give a proof of this result adapted to our particular case.

\begin{proposition}[Imagination strategy]
\label{prop:imagination}
There is no position $G$ of the \MB\ such that $o(G)=\mathcal P$.
\end{proposition}

\begin{proof}
Assume there is position $G$ of the \MB\ such that $o(G)=\mathcal P$. This means in particular that \Dom\ wins playing second on $G$. We next give a winning strategy for \Dom\ as first player, which will imply that $o(G)=\mathcal D$, a contradiction.

The strategy for \Dom as first player is the following. He first plays any vertex unplayed and then imagines he did not. He thus considers himself as the second player, seeing this vertex as an extra vertex.
Whenever his winning strategy (as a second player) requires to play the extra vertex, he plays any other unplayed vertex $u$, and considers $u$ as the new extra vertex.
If \Dom\ was winning before all the vertices were chosen, he still wins no later than his last move in the game where he was playing second.
Otherwise, when \St\ chooses the last vertex of the graph, her strategy asks her to play the extra vertex since it is the only one available in the imagined game, but it means that \Dom\ had already won on the previous turn.
\end{proof}

Note that this proposition is valid for any position of the game (and not only for starting positions). In other words, it ensures that a player has no interest to miss his/her turn. Figure~\ref{fig:ex_outcome} gives an example of graphs for the three remaining outcomes.\\

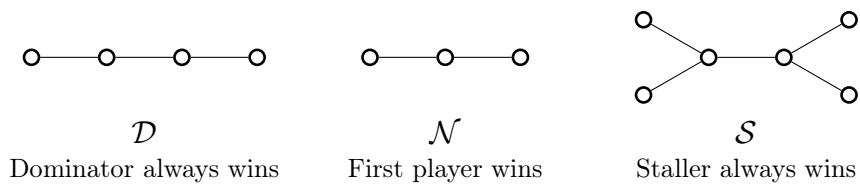
\begin{figure}[ht]
    \centering
    \begin{tikzpicture}
        
        \node at (-0.5,0){
            \begin{tikzpicture}
            	\node[noeud] (a) at (-1,0){}; 
				\node[noeud] (b) at (0,-0){};
				\node[noeud] (c) at (1,-0){};
				\node[noeud] (d) at (2,-0){};

				\draw (a) -- (b) -- (c)--(d);
            \end{tikzpicture}
        };
        
        \node at (3.5,0){
            \begin{tikzpicture}
            	\node[noeud] (a) at (-1,0){}; 
				\node[noeud] (b) at (0,-0){};
				\node[noeud] (c) at (1,-0){};

				\draw (a) -- (b) -- (c);
            \end{tikzpicture}
        };
        
        \node at (7.5,0){
            \begin{tikzpicture}
				\node[noeud] (b) at (0,0){};
				\node[noeud] (c) at (1,0){};
				\node[noeud] (d) at (-0.866,-0.5){};
				\node[noeud] (e) at (-0.866,0.5){};
				\node[noeud] (f) at (1.866,-0.5){};
				\node[noeud] (g) at (1.866,0.5){};
				
				\draw (c) -- (b);
				\draw (b) -- (d);
				\draw (b) -- (e);
				\draw (c) -- (f);
				\draw (c) -- (g);
            \end{tikzpicture}
        };
        
        \node at (-0.5,-1) {$\mathcal D$};
        \node at (3.5,-1) {$\mathcal N$};
        \node at (7.5,-1) {$\mathcal S$};
        \node at (-0.5,-1.5) {\footnotesize Dominator always wins};
        \node at (3.5,-1.5) {\footnotesize First player wins};
        \node at (7.5,-1.5) {\footnotesize Staller always wins};
        
    \end{tikzpicture}
    \caption{Example of a graph for each possible outcome.}
    \label{fig:ex_outcome}
\end{figure}


According to the three possible outcomes of a position, we now introduce an order relation on the outcomes derived from combinatorial game theory: $\mathcal{ S \prec N \prec D}$. This allows us to state the following proposition.

\begin{proposition}
\label{prop:subgraph}
Let $G=(V,E,c)$ be a position of the \MB\ and let $H=(V,E',c)$ be another position, with $E' \subseteq E$. Then $o(H) \preceq o(G)$.
\end{proposition}

\begin{proof} A reformulation of the proposition is that if \Dom\ has a winning strategy on $H$, then he also has a winning strategy on $G$.

Assume \Dom\ has a winning strategy on $H$.
A winning strategy for \Dom\ on $G$ is to apply the same strategy as on $H$. 
Indeed, for every possible sequence of moves of \St , \Dom\ is able to dominate $H$. Since every edge of $H$ is also in $G$, \Dom\ is also able to dominate $G$.
\end{proof}

In other words, adding edges to a position can only benefit \Dom, and removing edges can only benefit \St . Note that this property does not hold in the standard domination game.\\

Another result can be derived from Maker-Breaker games. The following theorem is a well known result from the early studies about positional games.

\begin{theorem}[Erd\H os-Selfridge Criterion~\cite{erdos-1973}]
\label{thm:criterion}
Given a Maker-Breaker game $G$ on an hypergraph $(X,\mathcal F)$, if
$ \sum_{A \in \mathcal F} 2^{-|A|}<\frac{1}{2}$ then Breaker wins on $G$ playing second.
\end{theorem}

In order to apply this theorem to the \MB , we need to consider a reverse version of it. Indeed, as the set $\mathcal F$ corresponds to the dominating sets of $G$, the sizes of the winning sets are not easy to control. Thus, we can also consider the \MB\ as the Maker-Breaker game where $\mathcal F$ is the set of the closed neighborhoods of every vertex of $G$. In that case, \Dom\ is the Breaker, and \St\ is the Maker. Now Theorem~\ref{thm:criterion} can be applied on this game:  

\begin{proposition}
Let $G$ be a starting position of the \MB\ and let $\delta$ be the minimum degree of $G$. If $|V|< 2^\delta$ then \Dom\ has a winning strategy for the \MB\ on $G$ playing second.
\end{proposition}

\begin{proof}
As stated before, the \MB\ on $G$ is a Maker-Breaker game played on $\mathcal H= (V, \mathcal{F})$ where $\mathcal{F}$ is the set of the closed neighborhoods of $G$, and \St\ plays the role of Maker in this game. Applying the Erd\H os-Selfridge Criterion, we know that if $\sum_{u \in \mathcal V} 2^{-|N[u]|}<\frac{1}{2}$ then \Dom\ has a winning strategy. For all $u$ in $V$, we have $N[u]\geq \delta +1$, hence $2^{-|N[u]|} \leq 2^{-(\delta+1)}$. Thus if $|V| \times 2^{-(\delta + 1)}<\frac{1}{2}$ then \Dom\ has a winning strategy. 
\end{proof}

This result can be applied to prove that some families of graphs are $\mathcal D$ (e.g. $r$-regular graphs having $r>\log_2 |V|$). In addition, it also suggests that highly connected graphs are more advantageous for \Dom . 

\section{Complexity}

In this section, we consider the computational complexity of deciding whether a game position of the \MB\ is $\mathcal{S}$, $\mathcal{N}$, or $\mathcal{D}$. First, remark that in the general case, deciding the outcome of a Maker-Breaker game $(X,\mathcal{F})$ is {\sc pspace}-complete. Indeed, this game exactly corresponds to the game {\sc pos-cnf} that was proved to be {\sc pspace}-complete in~\cite{poscnf}.\\

{\sc pos-cnf} is played on a formula $F$ in conjunctive normal form, with variables $X_1, \ldots,X_n$, where each variable is positive, that is $F = C_1 \wedge \cdots \wedge C_m$ with clauses $C_i = X_{i_1} \vee \cdots \vee X_{i_{k_i}}$. Two players, Prover and Disprover, alternate turns in choosing a variable that has not been chosen yet.
When all variables have been chosen, variables chosen by Prover are set to true, while variables chosen by Disprover are set to false.
Prover wins if $F$ is true under this valuation and Disprover wins otherwise.
Without loss of generality, we can consider that each variable appears in the formula, otherwise we consider the formula $F' = F \wedge (X_1 \vee \cdots \vee X_n)$. Clearly, any Maker-Breaker game $(X,\mathcal{F})$ is equivalent to a {\sc pos cnf} game, as $X$ corresponds to the set of variables, and the winning sets correspond to the clauses. Prover has the same role as Breaker, and Maker has the role of Disprover. \\

The complexity of this game remains {\sc pspace}-complete when reduced to instances of the \MB:

\begin{theorem}\label{thm:pspace}
Deciding the outcome of a \MB\ position is {\sc pspace}-complete on bipartite graphs.
\end{theorem}

\begin{proof}
We reduce the problem from {\sc pos-cnf}.
Let $F = C_1 \wedge \cdots \wedge C_m$ be a positive formula in conjunctive normal form using $n$ variables $X_1 \cdots X_n$. 

We build a bipartite graph $G = (V,E)$ from $F$ as follows. There is one vertex for each variable and two vertices for each clause:
$$V = \{x_i|1 \leq i \leq n\} \cup \{c^k_j|1 \leq j \leq m, 0 \leq k \leq 1\},$$
and an edge between a variable vertex $x_i$ and a clause vertex $c^k_j$ (with $k\in \{0,1\}$) if the variable $X_i$ appears in the clause $C_j$. Figure \ref{fig:reduction_POScNF} shows an example of such a construction, from the example where $F=(X_1 \vee X_2) \wedge (X_1  \vee X_4) \wedge  (X_2 \vee X_3 \vee X_4)$.\\

\begin{figure}[ht]
\centering
\scalebox{1}{
\begin{tikzpicture}

\node[noeud] (x1) at (0,0){};
\node[noeud] (x2) at (1.25,0){};
\node[noeud] (x3) at (2.5,0){};
\node[noeud] (x4) at (3.75,0){};

\node[above] at (x1) {$x_1$};
\node[above] at (x2) {$x_2$};
\node[above] at (x3) {$x_3$};
\node[above] at (x4) {$x_4$};

\node[noeud] (c1) at (-0.8,-3){};
\node[noeud] (c1b) at (0.1,-3){};
\node[noeud] (c2) at (1.5,-3){};
\node[noeud] (c2b) at (2.4,-3){};
\node[noeud] (c3) at (3.8,-3){};
\node[noeud] (c3b) at (4.7,-3){};

\node[below] at (c1){$c_1^0$};
\node[below] at (c1b) {$c_1^1$};
\node[below] at (c2) {$c_2^0$};
\node[below] at (c2b) {$c_2^1$};
\node[below] at (c3) {$c_m^0$};
\node[below] at (c3b) {$c_m^1$};

\draw (x1) -- (c1);
\draw (x1) -- (c1b);
\draw (x1) -- (c2);
\draw (x1) -- (c2b);
\draw (x2) -- (c1);
\draw (x2) -- (c1b);
\draw (x2) -- (c3);
\draw (x2) -- (c3b);
\draw (x3) -- (c3);
\draw (x3) -- (c3b);
\draw (x4) -- (c2);
\draw (x4) -- (c2b);
\draw (x4) -- (c3);
\draw (x4) -- (c3b);

\end{tikzpicture}
}
\caption{Reduction from {\sc pos-cnf} on $(X_1 \vee X_2) \wedge (X_1  \vee X_4) \wedge (X_2 \vee X_3 \vee X_4)$.}
\label{fig:reduction_POScNF}
\end{figure}
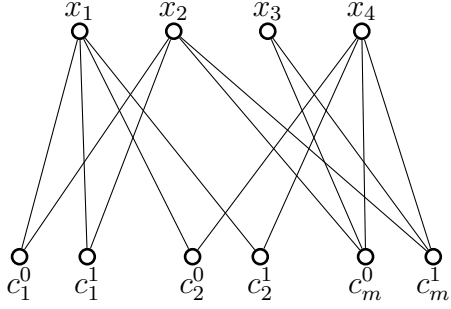

We now show that Prover has a winning strategy in $F$ as first player (respectively second player) if and only if \Dom\ has a winning strategy in $G$ as first (resp. second) player.

Assume Prover has a winning strategy in $F$. We first consider the case where Prover is the last player to play in {\sc pos-cnf} (i.e. $n$ is odd if Prover plays first and even if Prover plays second).

\Dom\ builds his strategy on $G$ as follows:
\begin{itemize}
\item If Prover and \Dom\ are starting the game, \Dom\ chooses the vertex $x_i$ corresponding to the variable $X_i$ played by Prover in his wining strategy.
\item Whenever \St\ chooses a vertex $c^k_j$, \Dom\ answers by choosing the vertex $c^{1-k}_j$.
\item Whenever \St\ chooses a vertex $x_i$, \Dom\ assumes Disprover chose the variable $X_i$. Then he answers by choosing the vertex $x_j$ corresponding to the variable $X_j$ played by Prover in his wining strategy.
\end{itemize}

This last step is always possible since we assume that Prover is playing the last in the {\sc pos-cnf} game. When all vertices are chosen, since Prover was winning in $F$, for each vertex $c^k_j$, there is a neighbor $x_i$ that was chosen by \Dom.
As all variables are in a clause, and for each $j$, \Dom\ chose either $c^0_j$ or $c^1_j$, all vertices of the form $x_i$ are also dominated by \Dom's choice of vertices.
Hence \Dom\ wins the game.

If Prover is not the last player to move, \Dom\ follows the same strategy but when \St\ is playing the last variable vertex, \Dom\ cannot answer a variable vertex. Then he can play any clause variable $c^k_j$ and imagines he did not, as in the Imagination strategy of Proposition~\ref{prop:imagination}, and goes on according to his strategy. If Staller answers the second vertex of the clause $C_j$ at some point, then \Dom\ chooses another unplayed clause vertex. At the end, we will also have, as before, one vertex of each clause plays by \Dom\ and the same conclusion holds.

Assume now Disprover has a winning strategy in $F$. The strategy for \St\ is exactly the same:
\begin{itemize}
\item Whenever Disprover's strategy requires to choose a variable $X_i$, \St\ chooses the vertex $x_i$.
\item Whenever \Dom\ chooses a vertex $c^k_j$, \St\ answers by choosing the vertex $c^{1-k}_j$.
\item Whenever \Dom\ chooses a vertex $x_i$, \St\ assumes Prover chose the variable $X_i$.
\end{itemize}

If the last step is not possible, this means that all the variables are chosen. Then, there exists a clause $C_j$ for which no variables are chosen by Prover. If $c^0_j$ and $c^1_j$ are already played, one of them has been chosen by \St, and thus is isolated. Otherwise, \St\ chooses $c^0_j$ and isolates it. In both cases, \St\ wins.
\end{proof}

\begin{corollary}
Deciding the outcome of a \MB\ position is {\sc pspace}-complete on chordal graphs, and also in particular on split graphs.
\end{corollary}

\begin{proof}
The proof of Theorem~\ref{thm:pspace} remains valid by adding edges between the variable vertices. In particular, if they form a clique, the resulting graph is a split graph, that is special case of chordal graphs.
\end{proof}
 
In view of these complexity results, the question of the threshold between {\sc pspace}-completeness and polynomiality is of natural interest. The following section is a first step towards it, with a characterization of a certain structure in the graph that induces a natural winning strategy for \Dom.


\section{Pairing strategy}

A natural winning strategy for Breaker in a Maker-Breaker game is the so-called {\em pairing strategy} as defined in \cite{hefetz2014positional}. This strategy can be applied when a subset of the board $X$ can be partitioned into pairs such that each winning set contains one of the pairs. In that case, a strategy for Breaker as a second player consists in occupying the other element of the pair that has been just occupied by Maker. By doing so, Breaker will occupy at least one element in each winning set and thus win the game. In the context of the \MB, such a subset correspond to a special dominating set that we introduce below.

\begin{definition}
Given a graph $G=(V,E)$, a subset of pair of vertices $\{(u_1,v_1),\ldots,(u_k,v_k)\}$ of $V$ is a {\em pairing dominating set} if all the vertices are distinct and if the intersection of the closed neighborhoods of each pair covers the vertices of the graph: $$V=\bigcup_{i=1}^k N[u_i]\cap N[v_i].$$
\end{definition}

Figure~\ref{fig:pair_dom_not_edge} shows an example of a pairing dominating set. Clearly, if one chooses a vertex in $(u_i,v_i)$ for each pair of a {\em pairing dominating set}, the resulting set is a dominating set of $G$.

\begin{figure}[ht]
\centering
\begin{tikzpicture}

\node[noeud] (1) at (-1.866,0.866){};
\node[noeud] (2) at (-1.866,-0.866){};
\node[noeud] (3) at (-1,0){};
\node[noeud] (4) at (0,0.866){};
\node[noeud] (5) at (0,0){};
\node[noeud] (6) at (0,-0.866){};
\node[noeud] (7) at (1,0){};
\node[noeud] (8) at (1.866,0.866){};
\node[noeud] (9) at (1.866,-0.866){};

\node[left] at (1){$u_1$};
\node[left] at (2){$v_1$};
\node[right] at (8){$u_2$};
\node[right] at (9){$v_2$};
\node[above] at (3){$u_3$};
\node[above] at (7){$v_3$};

\draw (3)--(1)--(2)--(3)--(4)--(7)--(8)--(9)--(7)--(5)--(3)--(6)--(7);

\end{tikzpicture}

\caption{The set $\{(u_1,v_1),(u_2,v_2),(u_3,v_3)\}$ is a pairing dominating set.}
\label{fig:pair_dom_not_edge}
\end{figure}
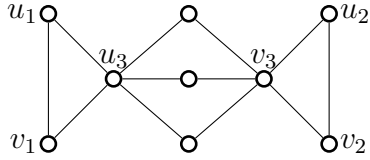

From this definition, we will say that a vertex $w$ is {\em pairing dominated} if there exists a pair $(u,v)$ from a pairing dominated set such that $w\in N[u]\cap N[v]$. In addition, all the pairs $(u,v)$ satisfying $N[u]\cap N[v]=\emptyset$ are useless in the construction of a pairing dominating set. Note that a pair $(u,v)$ of a pairing dominating set is not nessarily an edge of the graph.\\ 

The pairing strategy applied to the \MB\ can be translated into a strategy on a pairing dominating set:

\begin{proposition}\label{prop:pds}
If a graph $G$ admits a pairing dominating set, then $o(G)=\mathcal{D}$.
\end{proposition}

\begin{proof}
If $G$ admits a pairing dominating set, then \Dom\ applies the following strategy as a second player: each time \St\ occupies a vertex of a pair $(u_i,v_i)$ for some $i$, \Dom\ answers by occupying the other vertex of the same pair if it is not yet occupied. Otherwise, \Dom\ plays randomly. By definition of a pairing dominating set, it ensures that the vertices chosen by \Dom\ form a dominating set of $G$.
Hence \Dom has a strategy as second player, thus also as first player using Proposition \ref{prop:imagination}. 
\end{proof}

This result induces the following corollary that ensures a winning strategy for \Dom\ as a first player.

\begin{corollary}\label{cor:pds}
Given a graph $G$, if there exists a vertex $u$ of $G$ such that $G\setminus N[u]$ admits a pairing dominating set, then $\mathcal{N} \preceq o(G)$.
\end{corollary}
\begin{proof}
If such a vertex exists, then \Dom\ starts and occupy it. He then applies his pairing strategy on $G\setminus N[u]$ as a second player to dominate the rest of the graph.
\end{proof}

From this property, a natural question that arises is the detection of graphs having a pairing dominating set. An example of such graphs is when the vertices of the graph can be partitioned into cliques of size at least $2$. In that case, a trivial pairing dominating set consists in choosing any two vertices in each clique. Note that the question of the existence of such a partition is often referred to as the {\em packing by cliques} problem (with cliques of size at least $2$). It was proved to be polynomial by Hell and Kirkpatrick in \cite{hell}. A particular case of this decomposition is when the graph admits a perfect matching. As an example, Proposition~\ref{prop:pds} ensures that paths or cycles of even size are $\mathcal{D}$ as they have a perfect matching. \\

\begin{remark}
\label{rem:pairing}
The condition of Proposition~\ref{prop:pds} is not necessary. Indeed, the graphs of Figure~\ref{fig:D_no_pairing} are examples with outcome $\mathcal{D}$ and it can be shown that they do not admit a pairing dominating set. Yet, we will see in Section 5 two families of graphs (cographs and trees) for which there is an equivalence between the existence of a winning strategy for \Dom\ and the existence of a pairing dominating set.
\end{remark}

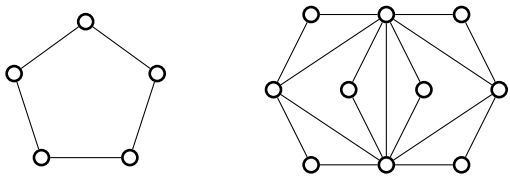
\begin{figure}[ht]
\centering
\begin{tikzpicture}

\node at (0,0){
\begin{tikzpicture}

    \foreach \thet in {18,90,162,234,306}{
        \node[noeud] (\thet) at (\thet:1){};
    }
    \draw (18) -- (90) -- (162) -- (234) -- (306) -- (18);
\end{tikzpicture}
};

\node at (4,0){
\begin{tikzpicture}
\tikzstyle{every node}=[noeud]
\node(1) at (-1.5,0) {};
\node(2) at (-0.5,0) {};
\node(3) at (0.5,0) {};
\node(4) at (1.5,0) {};
\node(5) at (-1,-1) {};
\node(6) at (0,-1) {};
\node(7) at (1,-1) {};
\node(8) at (-1,1) {};
\node(9) at (0,1) {};
\node(10) at (1,1) {};

\draw (1)--(5)--(6)--(1)--(8)--(9)--(1) (9)--(6)--(2)--(9)--(10)--(4)--(9)--(3)--(6)--(7)--(4)--(6);
\end{tikzpicture}
};

\end{tikzpicture}
\caption{Graphs with outcome $\mathcal D$ and without a pairing dominating set.}
\label{fig:D_no_pairing}
\end{figure}

We conclude this section with a study of the complexity of the pairing dominating set problem.

\begin{theorem}
\label{thm:pair_dom}
Given a graph $G$, it is {\sc np}-complete to decide whether $G$ admits a pairing dominating set.
\end{theorem}

\begin{proof}
Let $G=(V,E)$ be a graph. By definition, the problem is clearly in  {\sc np}. It remains to prove the  {\sc np}-hardness of the problem by reducing it from {\sc 3-sat}. Let $F= C_1 \vee \cdots \vee C_m$ be an instance of {\sc 3-sat} over the variables $X_1,\ldots,X_n$. Without loss of generality, one can assume that all the variables appear in both their positive and negative version in $F$, but not in the same clause. From $F$, we build the following graph $G$ as illustrated by Figure ~\ref{fig:gadget}.
\begin{itemize}
\item Each clause $C_j$, $1\leq j\leq m$, is associated to a vertex $c_j$. 
\item Each variable $X_i$, $1\leq i\leq n$ is associated to a gadget over seven vertices $\{x_i,y_i,z_i,x'_i,y'_i,z'_i,t_i\}$ such that $x_iy_iz_i$ and $x'_iy'_iz'_i$ are two triangles, and $t_i$ is adjacent to both $x_i$ and $x'_i$. The pairs $(x_i,y_i)$ and $(x'_i,y'_i)$ will be denoted $e_i$ and $\overline{e_i}$ respectively.
\item For each variable $X_i$ and clause $C_j$, we add the two edges $c_jx_i$ and $c_jy_i$ (resp. $c_jx'_i$ and $c_jy'_i$) if $X_i$ appears in clause $C_j$ in its positive (resp. negative) form.
\end{itemize}

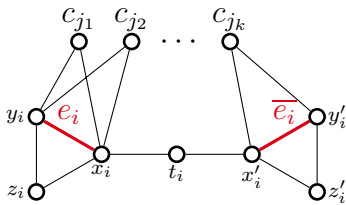
\begin{figure}[ht]
\centering
\begin{tikzpicture}

\node[noeud] (x) at (-1,0){};
\node[noeud] (y) at (-1.866,0.5){};
\node[noeud] (z) at (-1.866,-0.5){};

\node[below] at (x){\scriptsize $x_i$};
\node[left] at (y){\scriptsize $y_i$};
\node[left] at (z){\scriptsize $z_i$};

\node[noeud] (t) at (0,0){}; 

\node[below] at (t){\scriptsize $t_i$};

\node[noeud] (x') at (1,0){};
\node[noeud] (y') at (1.866,0.5){};
\node[noeud] (z') at (1.866,-0.5){};

\node[below] at (x'){\scriptsize $x'_i$};
\node[right] at (y'){\scriptsize $y'_i$};
\node[right] at (z'){\scriptsize $z'_i$};

\node[noeud] (c1) at (-1.3,1.5){};
\node[noeud] (c2) at (-0.6,1.5){};
\node[noeud] (c3) at (0.71,1.5){};

\node at (0.055,1.5) {$\cdots$};

\node[above] at (c1){$c_{j_1}$};
\node[above] at (c2){$c_{j_2}$};
\node[above] at (c3){$c_{j_k}$};

\draw[very thick, color = rougejoli] (x) -- (y) node[midway, above] {$e_i$};
\draw[very thick, color = rougejoli] (x') -- (y') node[midway, above] {$\overline{e_i}$};
\draw (y)--(z) -- (x) -- (t) -- (x') -- (z')--(y');
\draw (x) -- (c1) -- (y) --(c2) -- (x);
\draw (x') -- (c3) -- (y');

\end{tikzpicture}
\caption{Gadget around a variable $X_i$ for the proof of NP-completeness. The clauses $C_{j_1},\ldots,C_{j_k}$ are those where the variable $X_i$ appears.}
\label{fig:gadget}
\end{figure}

We first claim that any assignment of the variables $X_1,\ldots,X_n$ that makes $F$ satisfiable induces a pairing dominating set in $G$. Let $\sigma$ be such an assignment. We build the following set $D$ of pairs of vertices: for each variable $X_i$, we add the pairs $\{(x_i,y_i),(t_i,x'_i),(y'_i,z'i)\}$ to $D$ if $X_i$ is {\sc true} in $\sigma$ , and the pairs $\{(x'_i,y'_i),(t_i,x_i),(y_i,z_i)\}$ otherwise. It now suffices to check that $D$ is a pairing dominating set. First of all, one can easily remark that all the vertices of the gadgets (i.e., vertices different from the clauses $c_j$) are pairing dominated by $D$. In addition, as each clause $C_j$ is satisfied by $\sigma$, each vertex $c_j$ is adjacent to at least one pair $(x_i,y_i)$ or $(x'_i,y'_i)$ of $D$. Hence any choice of vertex in such a pair allows to dominate $c_j$. \\

Now consider a pairing dominating set $D$ of $G$. We first show that for each gadget associated a variable $X_i$, up to symmetry, there are only four cases to pairing dominate the vertices $t_i$, $z_i$ and $z'_i$, depicted by Figure~\ref{fig:pair_dom_gadget}. Indeed, since each vertex $t_i$ has degree $2$, there are three cases for it to be pairing dominated by $D$: either the pair $(t_i,x'_i)$, or $(t_i,x_i)$, or $(x_i,x'_i)$  must belong to $D$.

{\bf (i)} The pair $(t_i,x'_i)$ belongs to $D$. Then, by considering the vertex $z'_i$, of degree $2$, the pair $(y'_i,z'_i)$ must belong to $D$. Concerning the vertex $z_i$, it is necessarily dominated by vertices from the triangle $x_iy_iz_i$, leading to the three cases $(a)$, $(b)$ and $(c)$ of Figure~\ref{fig:pair_dom_gadget}.

{\bf(ii)} The pair $(t_i,x_i)$ belongs to $D$. By symmetry of the gadget, this case is similar to the previous one and we get the symmetric pairs from Figures $(a)$, $(b)$ and $(c)$.

{\bf(iii)} The pair $(x_i,x'_i)$ belongs to $D$ (Figure~\ref{fig:pair_dom_gadget} (d)). Then both vertices $z_i$ and $z'_i$ must belong to $D$ in the pairs $(y_i,z_i)$ and $(y'_i,z'_i)$.

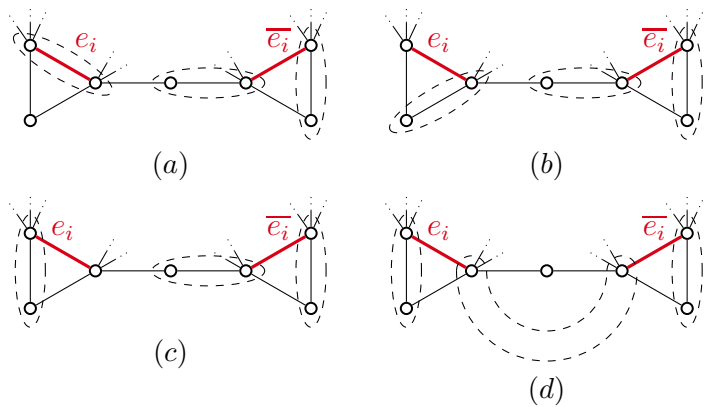
\begin{figure}[ht]
\centering
\begin{tikzpicture}

\node at (0,0){
\begin{tikzpicture}
\clip (-2.2,-1.9) rectangle (2.2,2.2);

\node[petit_noeud] (x) at (-1,0){};
\node[petit_noeud] (y) at (-1.866,0.5){};
\node[petit_noeud] (z) at (-1.866,-0.5){};

\node[petit_noeud] (t) at (0,0){}; 

\node[petit_noeud] (x') at (1,0){};
\node[petit_noeud] (y') at (1.866,0.5){};
\node[petit_noeud] (z') at (1.866,-0.5){};

\draw[very thick, color = rougejoli] (x) -- (y) node[midway, above right] {$e_i$};
\draw[very thick, color = rougejoli] (x') -- (y') node[midway, above] {$\overline{e_i}$};
\draw (y)--(z) -- (x) -- (t) -- (x') -- (z')--(y');
\draw[dashed] (1.866,0) ellipse (0.2 and 0.75);
\draw[dashed] (0.5,0) ellipse (0.75 and 0.2);
\begin{scope}[shift= {(-1,0)}]
    \begin{scope}[shift={(150:0.5)},rotate = 150]
        \draw[dashed] (0,0) ellipse (0.75 and 0.2);
    \end{scope}
\end{scope}

\begin{scope}[shift={(-1,0)}]
    \foreach \thet in {30,60}{
        \draw (x) --(\thet :0.3);
        \draw[dotted] (\thet :0.3) -- (\thet : 0.5);
    }
\end{scope}

\begin{scope}[shift={(-1.866,0.5)}]
    \foreach \thet in {65,90,125}{
        \draw (y) --(\thet :0.3);
        \draw[dotted] (\thet :0.3) -- (\thet : 0.5);
    }
\end{scope}

\begin{scope}[shift={(1,0)}]
    \foreach \thet in {120,150}{
        \draw (x') --(\thet :0.3);
        \draw[dotted] (\thet :0.3) -- (\thet : 0.5);
    }
\end{scope}

\begin{scope}[shift={(1.866,0.5)}]
    \foreach \thet in {65,90,125}{
        \draw (y') --(\thet :0.3);
        \draw[dotted] (\thet :0.3) -- (\thet : 0.5);
    }
\end{scope}

\node at (0,-1.1) {\small $(a)$};

\end{tikzpicture}
};

\node at (5,0){
\begin{tikzpicture}
\clip (-2.2,-1.9) rectangle (2.2,2.2);

\node[petit_noeud] (x) at (-1,0){};
\node[petit_noeud] (y) at (-1.866,0.5){};
\node[petit_noeud] (z) at (-1.866,-0.5){};

\node[petit_noeud] (t) at (0,0){}; 

\node[petit_noeud] (x') at (1,0){};
\node[petit_noeud] (y') at (1.866,0.5){};
\node[petit_noeud] (z') at (1.866,-0.5){};

\draw[very thick, color = rougejoli] (x) -- (y) node[midway, above] {$e_i$};
\draw[very thick, color = rougejoli] (x') -- (y') node[midway, above] {$\overline{e_i}$};
\draw (y)--(z) -- (x) -- (t) -- (x') -- (z')--(y');

\draw[dashed] (1.866,0) ellipse (0.2 and 0.75);
\draw[dashed] (0.5,0) ellipse (0.75 and 0.2);
\begin{scope}[shift= {(-1,0)}]
    \begin{scope}[shift={(-150:0.5)},rotate = -150]
        \draw[dashed] (0,0) ellipse (0.75 and 0.2);
    \end{scope}
\end{scope}

\begin{scope}[shift={(-1,0)}]
    \foreach \thet in {30,60}{
        \draw (x) --(\thet :0.3);
        \draw[dotted] (\thet :0.3) -- (\thet : 0.5);
    }
\end{scope}

\begin{scope}[shift={(-1.866,0.5)}]
    \foreach \thet in {65,90,125}{
        \draw (y) --(\thet :0.3);
        \draw[dotted] (\thet :0.3) -- (\thet : 0.5);
    }
\end{scope}

\begin{scope}[shift={(1,0)}]
    \foreach \thet in {120,150}{
        \draw (x') --(\thet :0.3);
        \draw[dotted] (\thet :0.3) -- (\thet : 0.5);
    }
\end{scope}

\begin{scope}[shift={(1.866,0.5)}]
    \foreach \thet in {65,90,125}{
        \draw (y') --(\thet :0.3);
        \draw[dotted] (\thet :0.3) -- (\thet : 0.5);
    }
\end{scope}

\node at (0,-1.1) {\small $(b)$};

\end{tikzpicture}
};

\node at (0,-2.5){
\begin{tikzpicture}
\clip (-2.2,-1.9) rectangle (2.2,2.2);

\node[petit_noeud] (x) at (-1,0){};
\node[petit_noeud] (y) at (-1.866,0.5){};
\node[petit_noeud] (z) at (-1.866,-0.5){};

\node[petit_noeud] (t) at (0,0){}; 

\node[petit_noeud] (x') at (1,0){};
\node[petit_noeud] (y') at (1.866,0.5){};
\node[petit_noeud] (z') at (1.866,-0.5){};

\draw[very thick, color = rougejoli] (x) -- (y) node[midway, above] {$e_i$};
\draw[very thick, color = rougejoli] (x') -- (y') node[midway, above] {$\overline{e_i}$};
\draw (y)--(z) -- (x) -- (t) -- (x') -- (z')--(y');

\draw[dashed] (-1.866,0) ellipse (0.2 and 0.75);
\draw[dashed] (0.5,0) ellipse (0.75 and 0.2);
\draw[dashed] (1.866,0) ellipse (0.2 and 0.75);

\begin{scope}[shift={(-1,0)}]
    \foreach \thet in {30,60}{
        \draw (x) --(\thet :0.3);
        \draw[dotted] (\thet :0.3) -- (\thet : 0.5);
    }
\end{scope}

\begin{scope}[shift={(-1.866,0.5)}]
    \foreach \thet in {65,90,125}{
        \draw (y) --(\thet :0.3);
        \draw[dotted] (\thet :0.3) -- (\thet : 0.5);
    }
\end{scope}

\begin{scope}[shift={(1,0)}]
    \foreach \thet in {120,150}{
        \draw (x') --(\thet :0.3);
        \draw[dotted] (\thet :0.3) -- (\thet : 0.5);
    }
\end{scope}

\begin{scope}[shift={(1.866,0.5)}]
    \foreach \thet in {65,90,125}{
        \draw (y') --(\thet :0.3);
        \draw[dotted] (\thet :0.3) -- (\thet : 0.5);
    }
\end{scope}

\node at (0,-1.1) {\small $(c)$};

\end{tikzpicture}
};

\node at (5,-2.5){
\begin{tikzpicture}
\clip (-2.2,-1.9) rectangle (2.2,2.2);
\node[petit_noeud] (x) at (-1,0){};
\node[petit_noeud] (y) at (-1.866,0.5){};
\node[petit_noeud] (z) at (-1.866,-0.5){};

\node[petit_noeud] (t) at (0,0){}; 

\node[petit_noeud] (x') at (1,0){};
\node[petit_noeud] (y') at (1.866,0.5){};
\node[petit_noeud] (z') at (1.866,-0.5){};

\draw[very thick, color = rougejoli] (x) -- (y) node[midway, above] {$e_i$};
\draw[very thick, color = rougejoli] (x') -- (y') node[midway, above] {$\overline{e_i}$};
\draw (y)--(z) -- (x) -- (t) -- (x') -- (z')--(y');

\draw[dashed] (-1.866,0) ellipse (0.2 and 0.75);
\draw[dashed] (1.866,0) ellipse (0.2 and 0.75);
\draw[dashed] (-1.2,0) arc (180:0:0.2);
\draw[dashed] (0.8,0) arc (180:0:0.2);
\draw[dashed] (-1.2,0) arc (-180:0:1.2);
\draw[dashed] (-0.8,0) arc (-180:0:0.8);

\begin{scope}[shift={(-1,0)}]
    \foreach \thet in {30,60}{
        \draw (x) --(\thet :0.3);
        \draw[dotted] (\thet :0.3) -- (\thet : 0.5);
    }
\end{scope}

\begin{scope}[shift={(-1.866,0.5)}]
    \foreach \thet in {65,90,125}{
        \draw (y) --(\thet :0.3);
        \draw[dotted] (\thet :0.3) -- (\thet : 0.5);
    }
\end{scope}

\begin{scope}[shift={(1,0)}]
    \foreach \thet in {120,150}{
        \draw (x') --(\thet :0.3);
        \draw[dotted] (\thet :0.3) -- (\thet : 0.5);
    }
\end{scope}

\begin{scope}[shift={(1.866,0.5)}]
    \foreach \thet in {65,90,125}{
        \draw (y') --(\thet :0.3);
        \draw[dotted] (\thet :0.3) -- (\thet : 0.5);
    }
\end{scope}

\node at (0,-1.6) {\small $(d)$};

\end{tikzpicture}
};

\end{tikzpicture}
\caption{Possible pair dominating sets for the gadget of the proof of Theorem~\ref{thm:pair_dom} (up to symmetry).}
\label{fig:pair_dom_gadget}
\end{figure}

In order to find an assignment for $F$, we now show that $D$ can be transformed into a pairing dominating set where each pair is as in Figure~\ref{fig:pair_dom_gadget} (a) (or its symmetrical, according to case $(ii)$). Consider first that for the gadget associated to some variable $X_i$, the pairs of $D$ are those depicted by Figure~\ref{fig:pair_dom_gadget} $(b)$. As the vertex $z_i$ has no other neighbor than $x_i$ and $y_i$, replacing a pair $(x_i,z_i)$ by the pair $(x_i,y_i)$ in $D$ remains a valid pairing dominating set since both $x_i$ and $y_i$ are adjacent to $z_i$. This operation is clearly possible if $y_i$ is not in $D$. In the case where $y_i$ is already in $D$, say in a pair $(y_i,u)$, remark that removing this pair from $D$ does not break the pairing dominating property of $D$ if $(x_i,y_i)$ is added. Indeed, since, by definition of $G$, $x_i$ and $y_i$ have the same neighborhood (except $t_i$, that is already in a pair), we have that $N[u]\cap N[y_i]\subseteq N[x_i]\cap N[y_i]$. Since $x_i$ and $y_i$ play a symmetrical role, we can use the same argument to replace the pairs of Figure~\ref{fig:pair_dom_gadget} $(c)$ by those of $(a)$ in $D$. The last case is when the pairs of $D$ are those of Figure~\ref{fig:pair_dom_gadget} $(d)$ for the variable $X_i$. Since $N[y_i]\cap N[z_i]\subseteq N[y_i]\cap N[x_i]$ and $N[x_i]\cap N[x'_i] = \{t_i\} \subset N[t_i]\cap N[x'_i]$ (as $X_i$ and $\overline{X_i}$ cannot be in the same clause), we can replace the pairs of Figure~\ref{fig:pair_dom_gadget} $(d)$ by those of Figure~\ref{fig:pair_dom_gadget} $(a)$ without breaking the pairing dominating property of $D$. In case $t_i$ was already in $D$, say in a pair $(t_i,u)$, once again this pair can be removed from $D$ as $N[t_i]\cap N[u]$ is either empty or at most a subset of $\{x_i,x'_i\}$, which is already pairing dominated by the pairs of Figure~\ref{fig:pair_dom_gadget} $(a)$.

Hence we have transformed $D$ such that all the vertices different from the $c_j$ are pairing dominated by the pairs of vertices of Figure~\ref{fig:pair_dom_gadget} $(a)$. In addition, if $D$ admits other pairs than those depicted by Figure~\ref{fig:pair_dom_gadget} $(a)$, then these pairs are necessarily of the form $(c_j,c_l)$, $(z_i,u)$, or $(z'_i,u)$. The last two types of pairs can be removed from $D$ as $N[z_i]$ and $N[z'_i]$ are already pairing dominated. Concerning the pairs $(c_j,c_l)$, they can also be removed from $D$ as the sets $N[c_j]\cap N[c_l]$ belong to the gadgets (and are different from the clause vertices), and are thus already pairing dominated. \\

We now build the following assignment of the variables of $F$: for all $1\leq i \leq n$, the variable $X_i$ is set to {\sc true} if and only if the pair $e_i$ belongs to $D$. As each vertex $c_j$ is pairing dominated in $D$ by at least a pair $e_i$ or $\overline{e_i}$ for some $i$, it means that each corresponding clause $C_j$ has at least a variable equal to {\sc true}, which concludes the proof.

\end{proof}

\section{Graph operations}

In the first part of this section, we study the outcome of operations of graphs for which the outcome is already known. This will lead to polynomial time algorithms to solve the \MB\ on cographs and forests, as these families can be built from joins, unions and by adjoining pendant edges.

\subsection{Union and join}

Let $G=(V_G,E_G)$ and $H=(V_H,E_H)$ be disjoint graphs. The \emph{union} $G \cup H$ of $G$ and $H$ is the graph with vertex set $V_G \cup V_H$ and edge set $E_G \cup E_H$.
The \emph{join} $G \bowtie H$ of $G$ and $H$ is the graph with vertex set  $V_G \cup V_H$ and edge set $E_G \cup E_H \cup \{uv|u\in V_G, v\in V_H\}$.

\begin{theorem}
\label{thm:union}
Let $G$ and $H$ be two starting positions of the \MB .
	\begin{itemize}
	\item If $o(G)=\mathcal S$ or $o(H)=\mathcal S$ then $o(G\cup H)=\mathcal S$.
	\item If $o(G)=o(H) = \mathcal N$ then $o(G\cup H)=\mathcal S$.
	\item If $o(G)=o(H) = \mathcal D$ then $o(G\cup H)=\mathcal D$.
	\item Otherwise, $o(G\cup H)=\mathcal N$.
	\end{itemize}
\end{theorem}

This result is summarized in Table~\ref{tab:union}.
Note that the outcome $\mathcal S$ is absorbing for the union, while the outcome $\mathcal D$ is neutral.

\begin{table}[ht]
\centering
\begin{tabular}{|c||c|c|c|}
\hline
\diagbox{$o(G)$}{$o(H)$}  & ${\cal D}$ & ${\cal N}$ & ${\cal S}$ \\ \hline\hline
${\cal D}$ & ${\cal D}$ & ${\cal N}$  & ${\cal S}$ \\ \hline
${\cal N}$ & ${\cal N}$ & ${\cal S}$ & ${\cal S}$ \\ \hline
${\cal S}$ & ${\cal S}$ & ${\cal S}$ & ${\cal S}$ \\ \hline
\end{tabular}
\caption{Outcomes of the \MB\ played on the union of $G$ and $H$.}
\label{tab:union}
\end{table}

\begin{proof}
Assume \St\ has a winning strategy on $G$ or $H$. Then she has a winning strategy on $G \cup H$. Indeed, without loss of generality assume that she has a winning strategy on $G$. Her strategy on $G\cup H$ is to play only on $G$ following her winning strategy. If at some point \Dom\ is playing on $H$, this can be considered as a passing move in $G$ and by Proposition~\ref{prop:imagination} this does not compromise Staller's strategy. At some point she will isolate a vertex in $G$ and thus in $G\cup H$.

Thus if $G$ or $H$ has outcome $\mathcal{S}$, then whatever \Dom\ plays as a first move, \St\ still has a winning strategy on this graph. If both positions have outcome $\mathcal{N}$ then after \Dom 's first move, \St\ can play on the other component and also wins. This proves the first two points.

If both positions have outcome $\mathcal D$, then \Dom\ has a winning strategy on both graphs playing second. He can answer to every move of \St\ on the component she plays with his winning strategy on this component. If one of the graph is full, \Dom can play any vertex on the other graph and imagines he did not, as in the imagination strategy of Proposition \ref{prop:imagination} At the end, \Dom\ dominates both components and so $G \cup H$ has outcome $\mathcal{D}$.

Finally, assume without loss of generality that $o(G)=\mathcal N$ and $o(H)=\mathcal D$. If \St\ plays first, as in the first case, by applying her winning strategy as the first player in $G$ she will be able to isolate a vertex and to win. On the other hand, if \Dom\ plays first, he can play his winning move on $G$ and then answers to \St\ on the component she has played on with his winning strategy. So the first player has a winning strategy and the outcome is $\mathcal N$.

\end{proof}

\begin{theorem}
\label{thm:join}
Let $G$ and $H$ be two starting positions of the \MB .
\begin{itemize}
\item[(i)] If $G=K_1$ and $o(H) = \mathcal S$  (or $H=K_1$ and $o(G) = \mathcal S$)\\ then $o(G \bowtie H)=\mathcal N$.
\item[(ii)] Otherwise, $o(G \bowtie H)=\mathcal D$.
\end{itemize}

\end{theorem}

\begin{proof}
(i) Assume that $G= K_1$ and $o(H) = \mathcal S$. If \Dom\ starts, he will win by playing on the unique vertex of $G$ and dominates the join, so he has a winning strategy as a first player. However, since $o(H)= \mathcal S$, if \St\ starts, she can play on the only vertex of $G$ and then apply her winning strategy as second player on $H$. So she wins on $G \bowtie H$ as first player as well as \Dom and $o(G \bowtie H)=\mathcal N$.

(ii) Since we are not in the first case, there are two possibilities : Either both $G$ and $H$ have at least two vertices or, without loss of generality, $G=K_1$ and $o(H) \succeq \mathcal N$. 

Assume first that both $G$ and $H$ have more than two vertices.
Let $u_1$, $v_1$ be two vertices of $G$ and $u_2$, $v_2$ two vertices of $H$. Since every vertex of $G$ is a neighbor of every vertex of $H$ and conversely, $\{(u_1,v_1),(u_2,v_2)\}$ forms a pairing dominating set for $G \bowtie H$ and the outcome is $\mathcal D$ according to Proposition~\ref{prop:pds}.

Assume now that $G=K_1$ and $o(H) \succeq \mathcal N$. Note that \Dom\ has a winning strategy on $H$ as first player.
Assume that \St\ is the first player. If on her first move she does not play on the vertex of $G$, then \Dom\ wins immediately by playing on it. If she does play on it, then \Dom\ will apply his winning strategy as first player on $H$. This will allow him to dominate $H$ and, since each vertex of $H$ dominates $G$, all the vertices of $G \bowtie H$ will be dominated. \Dom\ has a winning strategy as second player, hence $o(G \bowtie H) = \mathcal D$.
\end{proof}

The combination of these two results gives a complexity result on the class of cographs. Recall that cographs (or $P_4$-free graphs) can be inductively built from a single vertex by taking the union of two cographs or the join of two cographs. In addition, from a given cograph, recovering this construction from unions and joins can be found with a linear time algorithm~\cite{corneil1985linear}. Since we know the outcome of \MB\ for $K_1$ and for the union and the join operators, we can deduce the following corollary.

\begin{corollary}
Deciding the outcome of the \MB\ on cographs can be done in polynomial time.
\end{corollary}

As stated in Remark~\ref{rem:pairing}, for some families of graphs the outcome of a starting position is $\mathcal D$ if and only if it admits a pairing dominating set. We show that the family of cographs satisfies this property.

\begin{theorem}
A cograph $G$ has outcome $\mathcal D$ if and only if it admits a pairing dominating set.
\end{theorem}

\begin{proof}
By Proposition~\ref{prop:pds} that if a graph admits a pairing dominating set, then it has outcome $\mathcal D$. It remains to prove that all cographs with outcome $\mathcal D$ admits a pairing dominating set.

The proof is done by induction on the number $n$ of vertices of $G$.

First note that the result is true when $n \leq 2$. The only such cographs are $K_1$, $K_2$ and $K_1\cup K_1$, and among them the only graph with outcome $\mathcal D$ is $K_2$. $K_2$ admits a perfect matching and thus a pairing dominating set.

Assume now that every cograph of outcome $\mathcal D$ with a number of vertices less or equal to $n$ admits a pairing dominating set. Let $G$ be a cograph of outcome $\mathcal D$ with $n+1$ vertices. By definition of a cograph, $G$ is either the union or the join of two smaller cographs.

If $G$ is the union of two cographs $G_1$ and $G_2$, they necessarily have outcome $\mathcal D$ by Theorem~\ref{thm:union}. By induction hypothesis, they both admit a pairing dominating set, which union is a pairing dominating set for $G$.

Assume now that $G$ is the join of two cographs $G_1$ and $G_2$.

If both $G_1$ and $G_2$ have more than two vertices, then if $u_1$, $v_1$ are any two vertices of $G_1$ and $u_2$, $v_2$ are any two vertices of $G_2$, $\{(u_1,v_1),(u_2,v_2)\}$ forms a pairing dominating set for $G$.

Assume now that $G_1 = K_1$ and let $x$ be its unique vertex. Then $G_2$ has either outcome $\mathcal N$ or $\mathcal D$ by Theorem~\ref{thm:join}. If $G_2$ has outcome $\mathcal D$ then by induction hypothesis, it admits a pairing dominating set. Every vertex of this pairing dominating set is a neighbor of $x$ and it remains also a pairing dominating set for $G$.

Assume now that $o(G_2) = \mathcal N$. $G_2$ is either the union of two cographs or the join of two cographs.

If $G_2$ is the join of two cographs, by Theorem~\ref{thm:join}, it must be the join of a graph $K_1$ with vertex $y$ and of a graph $H$ with outcome $\mathcal S$. Notice that $x$ and $y$ are both universal vertices so $\{(x,y)\}$ is a pairing dominating set for $G$.
If $G_2$ is the union of $H_1$ and $H_2$ then, without loss of generality, by Theorem~\ref{thm:union} $o(H_1)=\mathcal D$ and $o(H_2)=\mathcal N$. By induction hypothesis, $H_1$ admits a pairing dominating set $S_1$. Note also that by Theorem~\ref{thm:join}, $x \bowtie H_2$ has outcome $\mathcal D$, so by induction hypothesis it admits a pairing dominating set $S_2$. Since $S_1$ pairing dominates $H_1$ and $S_2$ pairing dominates $\{x\} \cup H_2$, $S_1 \cup S_2$ forms a pairing dominating set for $G$.
\end{proof}

\subsection{Glue operator and trees}

We now study the operator consisting of gluing two graphs on a vertex. This operator will be useful in the study of trees. A more formal definition is the following:

\begin{definition}
Let $G=(V_G,E_G)$ and $H=(V_H,E_H)$ be graphs and let $u \in V_G$ and $v\in V_H$ be two vertices. The \emph{glued graph} of $G$ and $H$ at $u$ and $v$ is the graph $G \glue{u}{v} H$ with vertex set $(V_G \setminus \{u\})\cup (V_H \setminus \{v\}) \cup \{w\} $ (where $w$ is a new vertex) and for which $xy$ is an edge if and only if $xy$ is an edge of $G$ or $H$ or $y=w$ and $xu$ is and edge of $G$ or $xv$ is an edge of $H$.
\end{definition}

If the vertex $u$ is clear from the context or does not matter, the glue will be denoted by $G \glue{}{v} H$. Similarly if the vertex $v$ is not useful in the notation, we might also remove it.
Figure~\ref{fig:glue} gives a representation of the glued of two graphs.

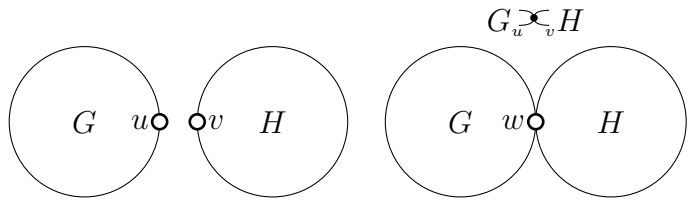
\begin{figure}[ht]
\centering

\begin{tikzpicture}

\draw (0,0) circle (1);
\draw (0,0) node {$G$};
\node[noeud] at (1,0) {};
\draw (1,0) node[left] {$u$};

\draw (2.5,0) circle (1);
\draw (2.5,0) node {$H$};
\node[noeud] at (1.5,0) {};
\draw (1.5,0) node[right] {$v$};

\draw (5,0) circle (1);
\draw (5,0) node {$G$};
\draw (7,0) circle (1);
\draw (7,0) node {$H$};
\draw (6,1.4) node {$G\glue{u}{v} H$};
\node[noeud] at (6,0) {};
\draw (6,0) node[left] {$w$};

\end{tikzpicture}

\caption{Representation of the glued graph of $G$ and $H$ on $u$ and $v$.}
\label{fig:glue}
\end{figure}

Let $H$ be a graph and $v$ a vertex of $H$. We say that the couple $(H,v)$ is \emph{neutral} for the glue operator if for every graph $G$ and every vertex $u$ of $G$, $o(G\glue{u}{v} H)=o(G)$.

\begin{theorem}
\label{thm:neutral}
Let $H$ be a graph and $v$ be a vertex of $H$. $(H,v)$ is neutral for the glue operator if and only if $o(H)=\mathcal N$ and $o(H \setminus \{v\})=\mathcal D$.
\end{theorem}

\begin{proof}
First, let $H$ be a graph and $v$ be a vertex of $H$. Assume that $(H,v)$ is neutral. Then $o(K_1\glue{}{v} H)=o(K_1)$. Notice that $K_1\glue{}{v}H=H$ and since $o(K_1)=\mathcal N$, we necessarily have $o(H)=\mathcal N$.

Now consider the graph $G=K_2 \glue{}{v} H$ that consists of $H$ with a pendant vertex $v'$ attached to $v$. Since $(H,v)$ is neutral, $G$ has the same outcome as $K_2$ that is $\mathcal D$. 
In particular, \Dom\ has a winning strategy on $G$ by playing second. If \St\ plays first on $v$, \Dom\ has to play on $v'$. His remaining winning strategy is a winning strategy on $H \setminus \{v\}$.

This proves that the conditions are necessary for $(H,v)$ to be neutral. We now prove that they are sufficient.

Let $H$ be a graph, $v$ be a vertex of $H$ and $H' = H \setminus \{v\}$, such that $o(H)= \mathcal N$ and $o(H') = \mathcal D$. Let $G$ be a graph and $u$ a vertex of $G$. In the following, we identify the vertices $u$ and $v$ to $w$ and the glued graph of $G$ and $H$ will be denoted by $G \glue{}{} H$.

Since $o(H')=\mathcal{D}$, $o(G \cup H')= o(G)$ by Theorem~\ref{thm:union}. Note that $G \cup H'$ is a subgraph of $G\glue{}{} H$ where only edges are removed so, by Proposition~\ref{prop:subgraph}, $o(G\glue{}{} H) \succeq o(G \cup H') = o(G)$.

We now show that $o(G\glue{}{} H) \preceq  o(G)$ to conclude the proof. Note that if $o(G) = \mathcal D$ we necessarily have $o(G\glue{}{} H) \preceq  o(G)$.

Assume that $o(G) \preceq \mathcal{N}$. This means that \St\ has a winning strategy on $G$ as first player. Since $o(H)=\mathcal N$, \St\ also has a winning strategy on $H$ as first player. The following strategy is a winning strategy on $G \glue{}{} H$ for \St\ as first player. \St\ begins by applying her winning strategy on $H$ until the strategy requires her to play on $w$. If during this stage \Dom\ plays on $w$, by following her strategy, \St\ will isolate a vertex on $H$ different from $w$. This vertex is not connected to $G$ so she wins. If \Dom\ plays on $G \setminus \{w\}$ then \St\ can imagine that \Dom\ has played on $w$ and will win similarly. So we can assume that \Dom\ always answers in $H'$.

When \St 's strategy on $H$ is to play on $w$, instead of playing $w$, she switches to her winning strategy on $G$. Similarly as before, if \Dom\ does not answer in $G \setminus\{w\}$, \St\ will win by isolating a vertex of $G$ different from $w$. Thus we can assume that \Dom plays only on $G \setminus \{w\}$. Then \St\ continues to apply her winning strategy on $G$ until this strategy requires her to play on $w$. Note that at this point $w$ is a winning move for \St\ both in $G$ and $H$.

\St\ now plays $w$ and answers to every move of \Dom\ with her strategy in the same component. Since she follows her winning strategy in $G$ and $H$ she will isolate a vertex in each of these graphs. If one of those two vertices is not $w$, then \St\ wins because this vertex is isolated in $G \glue{}{} H$. If both of these vertices are $w$, then $w$ and its whole neighborhood are played by \St\ in the glued graph and Staller wins.
So \St\ has a winning strategy as first player in $G\glue{}{} H$ and $o(G\glue{}{} H)\preceq \mathcal N$.

Assume now that $o(G)= \mathcal S$, i.e. \St\ has a winning strategy on $G$ as second player. If \Dom\ begins by playing on $w$, then \St\ can apply her winning strategy in $G$, she will isolate a vertex different from $w$ will win. If \Dom\ begins by playing in $H'$, then \St\ can imagine that he played on $w$, apply her winning strategy on $G$ and win similarly as before. So we can assume that \Dom\ begins by playing in $G \setminus \{w\}$. Then \St\ can follow the same strategy as before: she plays her winning strategy on $G$ until she wins or she has to play on $w$, when this is the case, she turns to her winning strategy as first player on $H$ until she wins or she has to play on $w$. Finally she isolates $w$. As before, \Dom\ has to answer to \St\ on the same graph \St\ has played. Thus $o(G\glue{}{} H) = \mathcal S$.

These three cases prove that $o(G\glue{u}{v} H) \preceq o(G)$. Since we also have $o(G\glue{u}{v} H) \succeq o(G)$, this prove that $o(G\glue{u}{v} H) = o(G)$.
\end{proof}

A question that could be asked is whether or not neutral graphs exist. We solve it by exhibiting an infinite family of neutral graphs:

\begin{definition}
For $n\geq 2$, the \emph{hanging split graph} of size $n$, $H_n$, is the graph composed of a clique of size $n$ with vertex set $\{v,v_1,\ldots , v_{n-1}\}$ and an independent of size $n-1$ with vertex set $\{u_1,\ldots,u_{n-1}\}$. Add an edge $u_iv_i$ for all $1 \leq i\leq n-1$.
\end{definition}

Figure~\ref{fig:glue_neutral} gives a representation of the first two hanging split graphs and of the general case.

\begin{proposition}
\label{prop:H_n}
For all $n \geq 2$, $(H_n, v)$ is neutral for the glue operator.
\end{proposition}

\begin{proof}
Note that $H_n \setminus \{v\}$ has a perfect matching so it has outcome $\mathcal D$ by Proposition~\ref{prop:pds}.

If \Dom\ plays first on $H_n$, a winning strategy is to start on $v$, then the remaining graph has a perfect matching and he will win.

If \St\ plays first on $H_n$, a winning strategy is to play on each $v_i$. \Dom\ has to answer on $u_i$ otherwise \St\ wins immediately by isolating this vertex. When every $v_i$ is played, she can play on $v$ and isolate it.

So both players have a winning strategy when playing first and thus $o(H_n) = \mathcal N$. By Theorem~\ref{thm:neutral}, $(H_n,v)$ is neutral.
\end{proof}

\begin{figure}[ht]
\centering
\begin{tikzpicture}

\node at (0,0){
\begin{tikzpicture}

\node[noeud] (a) at (-1,0){};
\node[noeud] (b) at (0,0){};
\node[noeud] (c) at (1,0){};

\draw (a)--(b)--(c);

\node[below right] at (1,0) {$v$};
\end{tikzpicture}
};

\node at (3,0){
\begin{tikzpicture}

\node[noeud] (a) at (0,0){};
\node[noeud] (b) at (60:1){};
\node[noeud] (c) at (120:1){};
\node[noeud] (d) at (60:2){};
\node[noeud] (e) at (120:2){};

\draw (d)--(b)--(a)--(c)--(e);
\draw (b)--(c);

\node[below right] at (0,0) {$v$};
\end{tikzpicture}
};

\node at (6,0){
\begin{tikzpicture}

\node at (0,0) {$K_n$};

\draw (0,0) circle (1);

\foreach \val in {0,60,120,240,300}{
    \node[noeud] (u\val) at (\val:1){};
}


\foreach \val in {60,120,240,300}{
    \node[noeud] (v\val) at (\val:1.5){};
    \draw (u\val) -- (v\val);
}

\node at (168:1.25) {\tiny $\bullet$};
\node at (180:1.25) {\tiny $\bullet$};
\node at (192:1.25) {\tiny $\bullet$};

\node[below right] at (1,0) {$v$};
\end{tikzpicture}
};

\node at (0,-2) {$H_2$};
\node at (3,-2) {$H_3$};
\node at (6,-2) {$H_n$};

\end{tikzpicture}
\caption{Examples of hanging split graphs.}
\label{fig:glue_neutral}
\end{figure}
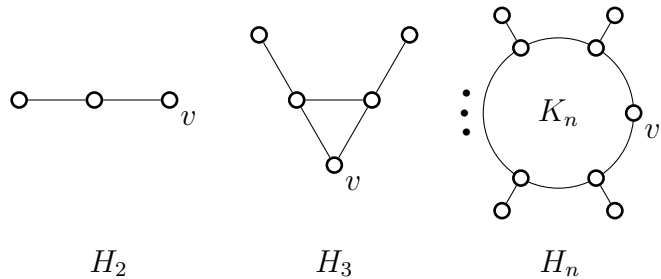

An interest of neutral graphs is that if a graph $G$ is of the form $G'\glue{}{v} H$ with $(H,v)$ being neutral, then we can restrict the study of $G$ to the study of $G'$. In the following, we apply this idea to trees by noticing that $P_3$ is isomorphic to $H_2$ and thus neutral.

We define a \emph{$P_2$-irreducible} graph as a graph without pendant $P_2$, where a pendant $P_2$ is a $P_2$ attached to a graph by an edge. Note that attaching a $P_2$ to a vertex is equivalent to gluing a $P_3$ to the same vertex.

\begin{lemma}
\label{lem:remove_p2}
Every $P_2$-irreducible tree has one of the following form :
\begin{itemize}
\item $K_1$
\item $P_2$
\item $K_{1,n}$ with $n \geq 3$
\item Trees where there are at least two vertices with more than two leaves as neighbors.
\end{itemize}
\end{lemma}

Figure~\ref{fig:red_trees} shows a representation of these different cases.

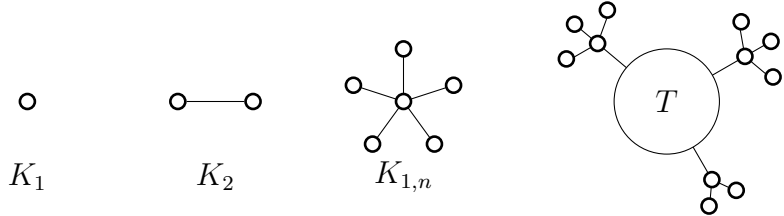
\begin{figure}[ht]
\centering

\begin{tikzpicture}
\draw (0,0) node[noeud] {};

\node[noeud] (a1) at (2,0){};
\node[noeud] (a2) at (3,0){};
\draw (a1) --(a2);

\draw (5,0) -- ++(18:0.7);
\draw (5,0) -- ++(90:0.7);
\draw (5,0) -- ++(162:0.7);
\draw (5,0) -- ++(234:0.7);
\draw (5,0) -- ++(306:0.7);

\draw (5,0) node[noeud] {};
\draw (5,0)++(18:0.7) node[noeud] {};
\draw (5,0)++(90:0.7) node[noeud] {};
\draw (5,0)++(162:0.7) node[noeud] {};
\draw (5,0)++(234:0.7) node[noeud] {};
\draw (5,0)++(306:0.7) node[noeud] {};

\draw (8.5,0) circle (0.7);
\draw (8.5,0) node {$T$};

\draw (8.5,0)++(140:0.7) -- ++(140:0.5);
\draw (8.5,0)++(30:0.7) -- ++(30:0.5);
\draw (8.5,0)++(-60:0.7) -- ++(-60:0.5);

\draw (8.5,0)++(140:1.2) -- ++(70:0.4);
\draw (8.5,0)++(140:1.2) -- ++(140:0.4);
\draw (8.5,0)++(140:1.2) -- ++(210:0.4);

\draw (8.5,0)++(30:1.2) -- ++(-40:0.4);
\draw (8.5,0)++(30:1.2) -- ++(30:0.4);
\draw (8.5,0)++(30:1.2) -- ++(100:0.4);

\draw (8.5,0)++(-60:1.2) -- ++(-25:0.4);
\draw (8.5,0)++(-60:1.2) -- ++(-95:0.4);

\draw (8.5,0)++(140:1.2) node[noeud] {};
\draw (8.5,0)++(140:1.6) node[noeud] {};
\draw (8.5,0)++(157:1.45)node[noeud] {};
\draw (8.5,0)++(123:1.45)node[noeud] {};
\draw (8.5,0)++(30:1.2) node[noeud] {};
\draw (8.5,0)++(30:1.6) node[noeud] {};
\draw (8.5,0)++(13:1.45) node[noeud] {};
\draw (8.5,0)++(47:1.45) node[noeud] {};
\draw (8.5,0)++(-60:1.2) node[noeud] {};
\draw (8.5,0)++(-68:1.5) node[noeud] {};
\draw (8.5,0)++(-52:1.5) node[noeud] {};

\draw (0,-1) node {$K_{1}$};
\draw (2.5,-1) node {$K_{2}$};
\draw (5,-1) node {$K_{1,n}$};

\end{tikzpicture}
\caption{Different possible reductions for trees.}
\label{fig:red_trees}
\end{figure}

\begin{proof}
Let $T$ be a $P_2$-irreducible graph.
If $T$ has only vertices of degree $1$ ans $2$, then $T$ is a path. The only paths that are $P_2$-irreducible are $K_1$ and $K_2$.

Thus, let $r$ be a vertex of degree at least 3 and consider $T$ as a rooted tree on $r$. Let $T_1$,...,$T_k$ be the subtrees connected to $r$.

Consider a subtree $T_i$ that is not a single vertex. Let $x_i$ be a leaf of maximal depth in $T_i$. The parent $y_i$ of $x_i$ has degree at least 3 (otherwise $T$ is not $P_2$-irreducible) and thus has at least one other child, which is necessarily a leaf by maximality of the depth of $x_i$.
Hence in every subtree that is not a single vertex, there is a vertex with at least two leaves as neighbours.

If there are two subtrees of size at least two, we are in the last case. If all the subtrees have size one, the tree is actually a star with at least three leaves. If there is exactly one subtree of size at least two, then $r$ is another vertex with at least two leaves (the other subtrees) and we are again in the last case.
\end{proof}

\begin{theorem}
Deciding the outcome of the \MB\ on trees is polynomial.
\end{theorem}

\begin{proof}
The following algorithm solves the \MB\ on trees in polynomial time:

For a tree $T$, iteratively remove a pendant $P_2$ until it is not possible anymore. Let $T'$ be the obtained tree. If $T'=P_2$, return the answer $\mathcal D$. If $T'=K_1$ or $K_{1,n}$ with $n \geq 3$, then return $\mathcal N$. Otherwise, return $\mathcal S$.

Note that the above algorithm is polynomial. Indeed, removing pendant $P_2$'s can be done in polynomial time by keeping in memory the set of leaves at each time and updating it when necessary. Verifying that a tree is $K_1$, $P_2$ or a star can also be done in polynomial time.

We now prove the correctness of the algorithm. Let $T_1, \ldots , T_k$ be the intermediary trees obtained after removing a pendant $P_2$. From Proposition~\ref{prop:H_n}, we know that $P_3$ is neutral and a pendant $P_2$ can be seen as the glue with a $P_3$. So $o(T) = o(T_1) = \ldots  = o(T_k) = o(T')$, and the outcome of $T$ is the same as the outcome of $T'$.
Since $T'$ is $P_2$-irreducible, it corresponds to one of the situations described in Lemma~\ref{lem:remove_p2}. If it is a $P_2$, the outcome is $\mathcal D$. If it is $K_1$ or $K_{1,n}$ with $n \geq 3$, the first player wins by playing on the central vertex and thus the outcome is $\mathcal N$. In the last case, two distinct vertices are attached to two leaves or more. Assume that \St\ plays second on $T'$. After \Dom 's first move, one of the two vertices and its leaves are unplayed by \Dom. \St\ can play this vertex and will isolate one of its leaves after her next move. Hence $T'$ is indeed $\mathcal S$ in this last case.

We conclude that the outcome of $T$ is the same as the outcome of $T'$ and the algorithm correctly returns the right output.
\end{proof}

\begin{remark}
Note that a tree has outcome $\mathcal D$ only if by removing pendant $P_2$'s the remaining graph is a $P_2$. This means that a tree has outcome $\mathcal D$ if and only if it admits a perfect matching and thus if and only if there is a pairing dominating set.
\end{remark}

\section{Conclusion and perspectives}

In this paper, the complexity of the \MB\ is studied for different classes of graphs. {\sc pspace}-completeness is proved for split and bipartite graphs, whereas polynomial algorithms are given for cographs and trees. An interesting equivalence property is that in these last two cases, the outcome is $\mathcal D$ if and only if the graph admits a pairing dominating set. The study of the pairing dominating set problem might be a key in the study of the threshold between {\sc pspace} and {\sc p} for the \MB .

As stated in the introduction, another problem that might be relevant to consider is the number of moves needed by \Dom\ to win. In particular, it could be worth studying the correlation of this value with the dominating number or the game dominating number. 

Also, this game has been built from the dominating set problem. Other remarkable structures in graphs could have been chosen, such as total dominating sets. Another variant would be to consider the game in an oriented version.

\section*{Acknowledgment}

The authors would like to thank Simon Schmidt, Milo\v s Stojakovi{\'c} and Sandi Klav\v zar for the fruitful discussions about this topic.

\end{document}